\documentclass{article}

\usepackage{times}
\usepackage{natbib}
\usepackage[margin=1.5in]{geometry}

\usepackage[utf8]{inputenc} 
\usepackage[T1]{fontenc}    
\usepackage{hyperref}       
\usepackage{url}            
\usepackage{booktabs}       
\usepackage{amsfonts}       
\usepackage{nicefrac}       
\usepackage{microtype}      
\usepackage{xcolor}         

\usepackage{amssymb}
\usepackage{amsthm,amsmath}
\usepackage{graphicx}

\usepackage{tikz-cd}
\usepackage{tikz-network}
\usepackage{caption}
\usepackage{mathdots}
\usepackage{csquotes}
\usepackage{bm}
\usepackage{xargs}
\usepackage{comment}
\usepackage{standalone}
\usepackage{sgame}
\usepackage[most]{tcolorbox}

\tikzset{ball/.style={circle, draw, fill=black,inner sep=0pt, minimum width=4pt}}
\tikzset{nd/.style={inner sep=1pt}}
\tikzset{CRS/.style={circle, draw,inner sep=1pt, minimum width=8pt}}

\usepackage{pgfplots}
\pgfplotsset{compat = newest}

\usetikzlibrary{decorations.markings}
\usetikzlibrary{arrows.meta}
\tikzset{>=Latex}
\tikzset{
  set arrow inside/.code={\pgfqkeys{/tikz/arrow inside}{#1}},
  set arrow inside={end/.initial=>, opt/.initial=},
  /pgf/decoration/Mark/.style={
    mark/.expanded=at position #1 with
    {
      \noexpand\arrow[\pgfkeysvalueof{/tikz/arrow inside/opt}]{\pgfkeysvalueof{/tikz/arrow inside/end}}
    }
  },
  arrow inside/.style 2 args={
    set arrow inside={#1},
    postaction={
      decorate,decoration={
        markings,Mark/.list={#2}
      }
    }
  },
}

\newtheorem{thm}{Theorem}[section] 

\newtheorem{defn}[thm]{Definition} 
\newtheorem{lem}[thm]{Lemma}
\newtheorem{corol}[thm]{Corollary}

\newcommand{\real}{\mathbb{R}}

\newcommand{\dist}{\mathbf{d}}
\newcommand{\sym}{\mathcal{S}}

\newcommand{\ddt}{\frac{\mathrm{d}}{\mathrm{d} t}}

\newcommand{\arc}[3][]{\begin{tikzcd} #2 \ar[r,-Latex,"#1"] \pgfmatrixnextcell #3 \end{tikzcd}}

\DeclareMathOperator{\supp}{support}

\DeclareMathOperator{\intr}{int}
\DeclareMathOperator{\content}{content}

\title{The Attractor of the Replicator Dynamic in Zero-Sum Games}

\date{}

\author{Oliver Biggar and Iman Shames}

\begin{document}

\maketitle

\begin{abstract}
    In this paper we characterise the long-run behaviour of the replicator dynamic in zero-sum games (symmetric or non-symmetric). Specifically, we prove that every zero-sum game possesses a unique global replicator attractor, which we then characterise. Most surprisingly, this attractor depends \emph{only} on each player's preference order over their own strategies and not on the cardinal payoff values, defined by a finite directed graph we call the game's \emph{preference graph}. When the game is symmetric, this graph is a tournament whose nodes are strategies; when the game is not symmetric, this graph is the game's response graph. We discuss the consequences of our results on chain recurrence and Nash equilibria.
\end{abstract}

\section{Introduction}

Learning in the presence of other learning agents is an increasingly fundamental topic in modern machine learning, motivated by its role at the core of cutting-edge techniques like \emph{learning from self-play}~\citep{silver2016mastering,silver2018general} and \emph{Generative Adversarial Networks}~\citep{goodfellow2020generative}. The challenge of these systems is analysing their collective behavior, which is where learning theory intersects with game theory. To quote~\citet{hofbauer1998evolutionary}, ``a major task of game theory [is] to describe the dynamical outcome of model games described by strategies, payoffs and adaptive mechanisms." That is, when agents learn collectively, what do they learn?

In online learning, the best-known approaches use variants of the Multiplicative Weights Update algorithm (MWU)~\citep{arora2012multiplicative}. To achieve the `no-regret' property, these algorithms typically decrease the step size as more samples are observed. In the long-run, as the step size becomes small, the behavior of a collection of MWU-playing agents converges to the flow of a famous differential equation: the \emph{replicator dynamic}~\citep{taylor1978evolutionary}. This model was originally inspired by biological models of evolution~\citep{smith1973logic}, and is a central object of study in \emph{evolutionary game theory}, the subfield of game theory which focuses on dynamic processes. Since its discovery, the replicator has been extensively analysed by biologists, mathematicians, economists and computer scientists~\citep{hofbauer2003evolutionary,sandholm2010population,kleinberg2011beyond}. Indeed, just as MWU is the flagship algorithm in online learning, the replicator is the flagship dynamic in evolutionary game theory~\citep{sandholm2010population}.

Describing the `dynamical outcome' of games under the replicator dynamic (and hence MWU) involves answering a basic question: \emph{to which strategy profiles do we converge over time?} In dynamical systems, a system's long-run behavior is defined by its \emph{attractors}~\citep{strogatz2018nonlinear}. Attractors are sets of points which are \emph{invariant} (points inside the set remain there for all time), \emph{asymptotically stable} (points in some neighbourhood converge to the set) and \emph{minimal} (they do not contain a smaller set with the first two properties). Understandably, identifying the attractors of the replicator is one of the core questions of evolutionary game theory, and increasingly also algorithmic game theory~\citep[see Section~\ref{sec: related}]{zeeman1980population,hofbauer2003evolutionary,papadimitriou2019game,omidshafiei2019alpha,vlatakis2020no}. The broadest takeaway of this line of work is that a strategy profile is an attractor \emph{if and only if} it is a \emph{strict (pure) Nash equilibrium}~\citep{sandholm2010population}. However, this result only describes the simplest attractors---those which contain only a single point. Most games---especially zero-sum games---don't have any pure Nash equilibria, and most attractors contain more than one point! Instead, the trajectories of the replicator in zero-sum games are typically periodic~\citep{mertikopoulos2018cycles}, and under MWU they are often \emph{chaotic}~\citep{cheung2019vortices,cheung2020chaos}.

We conclude that, despite four decades of research, \textbf{the attractors of the replicator remain largely unknown}, even in zero-sum games, arguably the best-studied special case~\citep[see Section~\ref{sec: related}]{akin1984evolutionary,hofbauer1996evolutionary,hofbauer1998evolutionary,hofbauer2003evolutionary,piliouras2014optimization}. This is what we achieve in this paper: we characterise the attractors of the replicator dynamic in every zero-sum game. Beyond this result, our concepts and techniques shed new light on equilibria, the graph structure of games~\citep{biggar2023graph,biggar2023replicator} and the modelling of payoffs/losses in economics and machine learning.

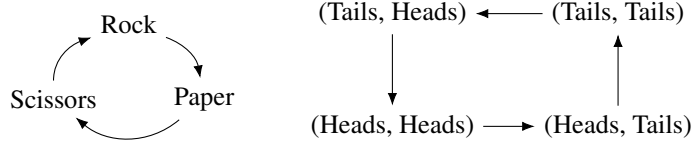
\begin{figure}
    \centering
    \begin{tikzpicture}

    \node (rock) at (0,1) {Rock};
    \node (paper) at (1,0) {Paper};
    \node (scissors) at (-1,0) {Scissors};
    
    \draw[->] (scissors) to[out=90,in=-160] (rock);
    \draw[->] (paper) to[out=-140,in=-40] (scissors);
    \draw[->] (rock) to[out=-20,in=100] (paper);
 \end{tikzpicture}
    \qquad
    \begin{tikzpicture}

    \node (hh) at (0,0) {(Heads, Heads)};
    \node (ht) at (3,0) {(Heads, Tails)};
    \node (th) at (0,1.5) {(Tails, Heads)};
    \node (tt) at (3,1.5) {(Tails, Tails)};
    
    \draw[->] (hh) to (ht);
    \draw[->] (ht) to (tt);
    \draw[->] (tt) to (th);
    \draw[->] (th) to (hh);
 \end{tikzpicture}
    \caption{The preference graphs (Definition~\ref{def: fundamental graph}) of two zero-sum games: (left) Rock-Paper-Scissors (symmetric), and (right) Matching Pennies (non-symmetric).}
    \label{fig:fundamental graphs}
\end{figure}

Characterising the replicator attractors of zero-sum games is a valuable development, but we believe the most remarkable aspect of this result is the form that this attractor takes. Specifically, the attractor depends \emph{only} on players' discrete preferences over their strategies, and not on the cardinal payoff values. These `preferences' are captured in a directed graph called the \emph{preference graph} of the game. In non-symmetric games this graph coincides with the game's \emph{response graph}~\citep{papadimitriou2019game}\footnote{We use the name \emph{preference graph} to unify the symmetric and non-symmetric cases, and because we find the name ``response graph" can be confusing. The word `response' suggests a \emph{repeated} or \emph{extensive-form} game, when actually the name is just a contraction of ``better-response relation", an ordering which defines player's \emph{preferences}.}, an object which has been of increasing interest in algorithmic game theory, particularly in relation to the replicator~\citep{candogan2011flows,omidshafiei2019alpha,biggar2023graph,biggar2023replicator}. The nodes of the preference graph are the profiles of the game, and the arcs represent which strategies players prefer, given the strategies of the other player. As an example, consider Figure~\ref{fig:fundamental graphs}, which shows the preference graph of the Matching Pennies game. In this game, player 1 prefers to match the choice of player 2, and player 2 prefers to mismatch player 1. The arc $\arc{(T,H)}{(H,H)}$\!\!, for example, captures the fact that, given player 2 plays Heads, player 1 `prefers' Heads over Tails. In symmetric zero-sum games, like Rock-Paper-Scissors (Figure~\ref{fig:fundamental graphs}), the preference graph has an even simpler form where each arc represents the preferred option between some pair of strategies. For example, given Rock `beats' Scissors, in a match-up of Rock and Scissors, one prefers to play Rock, hence the arc $\arc{\text{Scissors}}{\text{Rock}}$\!\!. \textbf{Conceptually, the preference graph stores the underlying combinatorial structure of the game.} Most game-theoretic concepts---including the replicator dynamic and the Nash equilibrium---are defined by cardinal payoffs, which serve as a numerical instantiation of the underlying preference structure. Our result shows that the choice of representation of preferences by numbers has a transient effect: two games with different payoffs but the same preferences have the same long-run behavior, in that their attractors are identical. One practical consequence is that computing the attractor is easy (we can do it by traversing the preference graph). More fundamentally, this lends our prediction stability in the face of uncertainty in our model, an important and rare property in game theory~\citep{von2007theory}.

\subsection{Contributions}

\begin{figure}
    \centering
    \begin{tikzpicture}
\def \ca {2.6,2}
\def \cb {0.5,1.3}
\def \cc {0,0}
\def \ba {2,3.3}
\def \bb {0,2.6}
\def \bc {-0.5,1.3}
\def \aa {0,4}
\def \ab {-2,3.3}
\def \ac {-2.6,2}

\node at (-7,2) {
\begin{game}{3}{3}[Player~1][Player~2]
         & $a$ & $b$ & $c$ \\
         $a$ & $7,-7$ & $6,-6$ & $3,-3$ \\
         $b$ & $8,-8$ & $10,-10$ & $0,0$ \\
         $c$ & $9,-9$ & $0,0$ & $10,-10$ 
         \end{game}
};

    \fill [gray!20] (\cc) to [in = -104, out = 7] (\ca) to (\cb) -- cycle;
    \fill [gray!20] (\cc) to [in = -104, out = 7] (\ca) to (\ba) to [in = 4, out = -107] (\bc) -- cycle;
    \fill [gray!20] (\cb) -- (\ca) -- (\ba) -- (\bb) -- cycle;
    \fill [gray!20] (\ba) to [in = 4, out = -106] (\bc) -- (\bb) -- cycle;
    \fill [gray!20] (\cc) to (\cb) to (\bb) to (\bc) -- cycle;
    \fill [gray!20] (\cc) to [out = 176, in = -73] (\ac) to (\bc) -- cycle;
    \fill [gray!20] (\cc) to [out = 176, in = -73] (\ac) to (\ab) to [out = -76, in = 173] (\cb) -- cycle;
    \fill [gray!20] (\ab) to [out = -76, in = 173] (\cb) -- (\bb) -- cycle;
    \fill [gray!20] (\bc) to (\ac) -- (\ab) -- (\bb) -- cycle;

    \node[inner sep=1pt] (ca) at (\ca) {$c,a$};
    \node[inner sep=1pt] (cb) at (\cb) {$c,b$};
    \node[inner sep=1pt] (cc) at (\cc) {$c,c$};
    \node[inner sep=1pt] (ba) at (\ba) {$b,a$};
    \node[inner sep=1pt] (bb) at (\bb) {$b,b$};
    \node[inner sep=1pt] (bc) at (\bc) {$b,c$};
    \node[inner sep=1pt] (aa) at (\aa) {$a,a$};
    \node[inner sep=1pt] (ab) at (\ab) {$a,b$};
    \node[inner sep=1pt] (ac) at (\ac) {$a,c$};

    \draw[->] (ba) to (ca);
    \draw[->] (cc) to [in = -110, out = 10] (ca);
    \draw[->] (ac) to [in = 170, out = -70] (cc);
    \draw[->] (ab) to (ac);
    \draw[->] (ab) to (bb);
    \draw[->] (bb) to (ba);
    \draw[->] (aa) to (ba);
    \draw[->] (aa) to [in = 170, out = -70] (ca);
    \draw[->] (aa) to (ab);
    \draw[->] (aa) to [in = 10, out = -110] (ac);
    \draw[->] (ca) to (cb);
    \draw[->] (cc) to (cb);
    \draw[->] (cb) to (bb);
    \draw[->] (cb) to [in = -70, out = 170] (ab);
    \draw[->] (bc) to (cc);
    \draw[->] (bc) to (ac);
    \draw[->] (bb) to (bc);
    \draw[->] (ba) to [in = 10, out = -110] (bc);
    

 \end{tikzpicture}
    \caption{A zero-sum game (left) and its associated preference graph (right). The sink component of the graph consists of all profiles other than $(a,a)$. By Theorem~\ref{attractor characterisation}, this game's unique replicator attractor is the \emph{content} (Definition~\ref{def: content}) of the sink component, which is the union of the strategy spaces of the subgames $\{a,b,c\}\times \{b,c\}$ and $\{b,c\}\times\{a,b,c\}$, represented by the shaded region on the graph. Note that the strategy space of the game is 4-dimensional, with the attractor a 3-dimensional region on the boundary.} 
    \label{fig:outer diamond}
\end{figure}
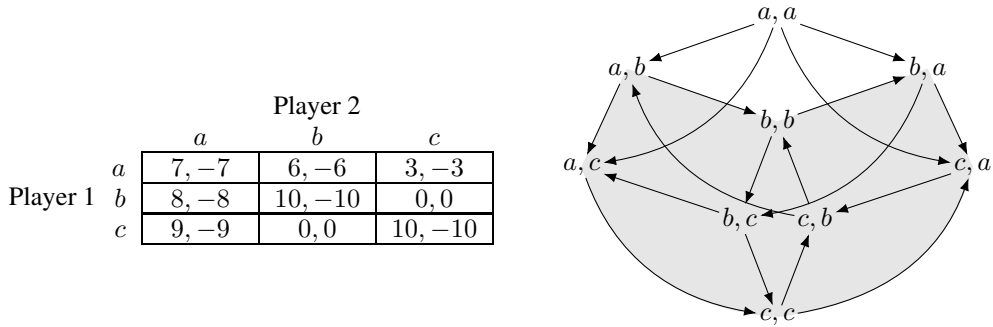

The main result of the paper is Theorem~\ref{attractor characterisation}, which characterises the attractors of zero-sum games. For each zero-sum game, we prove that an attractor exists, is unique and attracts all points on the interior of the game. This attractor is precisely the \emph{content}~\citep[Definition~\ref{def: content}, see][]{biggar2023replicator} of the preference graph's unique\footnote{Uniqueness of the sink connected component is a property of zero-sum games~\citep{biggar2023graph}. Non-zero-sum games, such as the $2\times 2$ Coordination game, may have preference graphs with multiple sink components, and thus can have multiple attractors under the replicator dynamic~\citep{biggar2023graph,biggar2023replicator}.} \emph{sink connected component}, which is a strongly connected component with no arcs from a node inside the component to a node outside. The sink component is a set of pure profiles; its content is the set of \emph{mixed} profiles whose support contains only profiles in this component, which is always an invariant set under the replicator. See Figure~\ref{fig:outer diamond}. A recent result~\citep{biggar2023replicator} demonstrated that every replicator attractor contains the content of some sink component. The challenging part of our proof is showing that the content of the sink component is asymptotically stable in any zero-sum game, and so is itself an attractor. We demonstrate stability using a potential function argument. Our choice of function derives from the preference graph: specifically, we use the total probability mass over all sink component profiles. The proof then separates into two cases, reflecting two standard types of zero-sum game: symmetric and non-symmetric. In evolutionary game theory these are often called \emph{single-}\footnote{Single-population games are also called \emph{population games}~\citep{sandholm2010population} or \emph{matrix games}~\citep{andrade2021learning}.} and \emph{multiple-population} games. The replicator has different properties in each case. In symmetric zero-sum games the preference graph possesses a simplified form, which makes the remainder of the proof straightforward. The non-symmetric case is much more complex. Here we prove a novel transformation of the replicator dynamic (Theorem~\ref{thm: symmetrisation}): the flow of the two-population replicator on a zero-sum game embeds in the flow of the single-population replicator on a larger symmetric zero-sum game, known as its \emph{von Neumann symmetrisation} (Definition~\ref{def:symised game}).
We believe this result (which also works for MWU) to be of independent interest. While symmetrisation sets up the proof, a further step is required because the sink component of the von Neumann symmetrisation may be larger than that of the original game. Lemma~\ref{lem: two-pop lyapunov} completes the proof using the fact that the dynamics are embedded on the subspace of product distributions. In Section~\ref{sec: nash} we discuss some consequences of our result. Lemma~\ref{lem: preference nash} shows an interesting connection between the preference graph and Nash equilibria in zero-sum games: the support of the equilibrium must be strongly connected as a subgraph of the preference graph and wholly contained within the unique sink connected component. This has important consequences for characterising \emph{sink chain components} (see Section~\ref{sec: related}).

\section{Related Work} \label{sec: related}

Long-run stability of strategy profiles under the replicator is a fundamental topic in evolutionary game theory, especially with regard to classical solution concepts, such as Nash equilibria and \emph{evolutionarily stable strategies}. See \citet{sandholm2010population} for a summary. A particularly relevant work is that of \citet{eshel1983coevolutionary}, who studied zero-sum games under the replicator, proving a crucial \emph{volume conservation} property, analysed in depth by \citet{hofbauer1996evolutionary}. Many papers since~\citep[such as][and this paper]{piliouras2014optimization,vlatakis2020no,biggar2023replicator} have used this property to bound asymptotically stable sets of the replicator. \citet{zeeman1980population} performed an early study of replicator attractors, suggesting that the qualitative behaviour of the replicator can be split into a finite number of classes; we extend this by showing that in zero-sum games the qualitative behaviour is defined by the preference graph alone. \citet{ritzberger1995evolutionary} showed that when a \emph{subgame} is closed under ``weakly better responses", then it is asymptotically stable under the replicator. The preference graph is defined by the weakly better responses, and so this follows as a special case of our result when a connected component of the preference graph is a subgame. Finally, one of the most famous of these classical results proves that the \emph{time-average} of the replicator (and MWU) converges to the Nash equilibrium in zero-sum games~\citep{freund1999adaptive,hofbauer2009time}. However, the time-average behavior is distinct from the day-to-day or \emph{last-iterate} behavior~\citep{papadimitriou2019game}.

Despite these efforts, a general negative conclusion of this line of work is that the replicator doesn't converge to mixed equilibria~\citep{sandholm2010population}, and moreover no dynamic can converge to equilibria in all games~\citep{hart2003uncoupled,benaim2012perturbations}. Instead, its behavior can be chaotic~\citep{sato2002chaos}. Further, finding equilibria is generally computationally intractable~\citep{daskalakis2009complexity}. Consequently, the algorithmic game theory and learning communities have increasingly shifted towards new notions of dynamical outcome which can predict the day-to-day behavior of computational agents in games~\citep{kleinberg2011beyond,papadimitriou2019game}.

 To this end, \emph{sink chain components} were recently proposed~\citep{papadimitriou2019game} as the outcome of dynamic games, with the replicator used as the motivating example. Sink chain components are built on a concept called \emph{chain recurrence} (Definition~\ref{def: chain recurrence}), a generalisation of periodicity which forms the foundation of the Fundamental Theorem of Dynamical Systems~\citep{conley1978isolated}. Crucially, chain components are grounded in computational considerations. To quote~\citeauthor{papadimitriou2019game}, informally, a ``point $x$ is chain recurrent if, whenever Alice starts at $x$, Bob can convince her that she is on a cycle by manipulating the round off error of her computer---no matter how much precision Alice brings to bear.'' Interestingly, this computational solution concept brings us back to the classical notion of an attractor: when a replicator attractor exists, it is a sink chain component (Lemma~\ref{attractors are sink chain components}). Thus, finding the attractors of the replicator dynamic is motivated not only by dynamical systems but also computer science: the attractors give us the strongest prediction of long-run behavior which is consistent with reliable computation.

This chain recurrence approach has inspired a number of new results on games and the replicator dynamic~\citep{omidshafiei2019alpha,biggar2023graph,biggar2023replicator}. In zero-sum games, when a \emph{fully-mixed} Nash equilibria exists, the behavior is essentially unpredictable: the sink chain component is the whole game~\citep{papadimitriou2016nash,papadimitriou2018nash}. Under MWU, we observe chaotic behavior in these games~\citep{cheung2019vortices,cheung2020chaos}. Further, when a fully-mixed NE does not exist, \citet{piliouras2014optimization} showed that all fully-mixed strategy profiles converge to the subgame containing the equilibrium, called the \emph{essential subgame}, and within this subgame all profiles are chain recurrent~\citep{papadimitriou2016nash}. Surprisingly, however, these results do not characterise the chain components of zero-sum games---the essential subgame is generally \emph{not} a sink chain component. The reason is that the convergence to the essential subgame is not \emph{uniform}, and so the essential subgame is typically \emph{not} asymptotically stable (see Section~\ref{sec: nash}), one of the defining properties of attractors and sink chain components~\citep{alongi2007recurrence} and a key property for predicting learning processes~\citep{vlatakis2020no,omidshafiei2019alpha}. Instead, interior profiles which are \emph{arbitrarily close} to the essential subgame may move far away before returning. The question remains unanswered: what are the replicator sink chain components/attractors of zero-sum games?


The story of predicting the replicator seems generally negative: the replicator may be chaotic, may not converge, or may only converge in time-average or non-uniformly. Our results tell a different, more positive story: we characterise the attractor/sink chain component of the replicator, which, while larger than the essential subgame, is the smallest outcome which is consistent with computation, in the sense of chain recurrence.  What's more, being defined by discrete preferences, the solution is natural and robust. As an example, in Figure~\ref{fig:outer diamond}, shifting the payoff for the profile $(a,c)$ from $(3,-3)$ to $(4,-4)$ moves the essential subgame from $\{b,c\}\times\{b,c\}$ to $\{a,c\}\times\{b,c\}$. However, all preferences remain unchanged, so the graph and hence the attractor do not change. Shifting focus from the Nash equilibrium also gives a new perspective on chain recurrence. Prior approaches~\citep{papadimitriou2016nash,mertikopoulos2018cycles} suggested a connection between equilibria and chain recurrence in zero-sum games. We find instead that chain recurrence in zero-sum games is \emph{entirely} defined by the preference graph (Lemma~\ref{sink chain components}). The previous findings are now explained by a non-trivial connection between equilibria and the preference graph: the existence of a fully-mixed equilibrium implies strong connectedness of the preference graph (Lemma~\ref{lem: preference nash}). See Section~\ref{sec: nash}.

\section{Preliminaries} \label{sec: preliminaries}

In game theory, a game is defined by a triple consisting of the \emph{players}, \emph{strategy sets} for each player, and \emph{payoffs}. A combination of strategies for each player is called a \emph{strategy profile} or simply a profile, and for each profile there is a real-valued payoff to each player. In this paper we focus on two-player games, where we denote the players by the integers 1 and 2 and their strategy sets by $S_1$ and $S_2$. The strategy names are simply labels, so we assume $S_1 = [n] := \{i\in \mathbb{N}_0 |\ i < n\}$ and $S_2 = [m] := \{i\in \mathbb{N}_0 |\ i < m\}$. The profiles are the pairs $S_1\times S_2$. We call this an $n\times m$ game. An $n\times m$ game is defined by a pair $(A,B)$ of matrices, $A\in\real^{n\times m}$ and $B\in\real^{m\times n}$, representing the payoffs to players 1 and 2 respectively. A game is \emph{symmetric} if $A = B$. Intuitively, a game is symmetric if the payoff is determined solely by the choice of strategies, and not identity of the player. We focus on \emph{zero-sum games}, which we represent by a single matrix $M\in \real^{n\times m}$, implicitly assumed to be the payoffs for the first player, which defines the game $(M,-M^T)$. That is, in a profile $(s_1,s_2)$, player 1 receives $M_{s_1,s_2}$ and player 2 receives $(-M^T)_{s_2,s_1} = -M_{s_1,s_2}$. A zero-sum game is also a symmetric game if $M = -M^T$, that is, $M$ is \emph{anti-symmetric}. Thus, there is a natural one-to-one correspondence between anti-symmetric real matrices and symmetric zero-sum games.
A \emph{subgame} of a game is formed by choosing subsets $T_1\subseteq S_1$ and $T_2\subseteq S_2$ of each player's strategy sets and restricting the game to the profiles in $T_1\times T_2$. We typically represent a subgame by its product set of profiles $T_1\times T_2$.

A \emph{mixed strategy} is a distribution over a player's strategies, and a \emph{mixed profile} is an assignment of a mixed strategy to each player. We sometimes refer to a profile as a \emph{pure profile} to distinguish it from a mixed profile. If $x$ is a mixed strategy, and $s$ a strategy, we write $x_s$ for the $s$-entry of $x$. Like profiles, we denote mixed profiles by pairs $(x,y)$ where $x$ and $y$ are mixed strategies for the first and second player, respectively. The \emph{support} of a mixed strategy $x$, written $\supp(x)$, is the set of strategies $s$ where $x_s$ is non-zero. The \emph{support} of a mixed profile $(x,y)$ is the Cartesian product $\supp(x)\times \supp(y)$, the set of profiles whose strategies are in the support of $x$ and $y$ respectively. As distributions over a finite set, mixed strategies can be naturally embedded in the standard probability simplex in Euclidean space, by choosing some arbitrary ordering of the strategies in $S_1$ and $S_2$. We denote these spaces by $\Delta(S_1)$ and $\Delta(S_2)$. The set of mixed profiles is the product $\Delta(S_1)\times \Delta(S_2)$, which we call the \emph{strategy space} of the game, and we refer to mixed profiles as `points' in strategy space. The strategy space is also naturally embedded in Euclidean space, so we can talk about geometric properties of sets of mixed profiles. The payoffs extend to mixed profiles using expectation. In a zero-sum game $(M,-M^T)$, the \emph{expected payoff} of a mixed profile $(x,y)$ is $y^T M x$ to player 1 and $x^T (-M^T) y = -(y^T M x)$ to player 2.

In evolutionary game theory, we think of an $n\times m$ game as a pair of `populations', with the strategies as `types' within each population. A mixed strategy $x$ represents the distribution of types in the population. Symmetric games, called \emph{single-population games} in this context, are viewed as having a a single underlying population, with types from the strategy set $S$. See~\cite{sandholm2010population}. Viewed as a two-player game, all mixed profiles in a single-population game are of the form $(x,x)$, because there is only one underlying population. This leads to an important terminology convention: \emph{in symmetric games, (mixed) profiles are the same as (mixed) strategies}. Intuitively, this is because a profile is a choice of strategy for each player, and symmetric games have `only one player'. The strategy space of the game consists of the symmetric profiles $(x,x)$ in $\Delta(S)\times \Delta(S)$. We write the mixed profile $(x,x)$ as the mixed strategy $x$, and write simply $\Delta(S)$ for the strategy space of the symmetric game.

\subsection{Preference graphs}

Two profiles are \emph{$i$-comparable} if they differ only in the strategy of player $i$; they are \emph{comparable} if they are $i$-comparable for some player $i$. If two profiles are comparable, then there is exactly one $i$ such that they are $i$-comparable. 
In symmetric games, profiles and strategies are the same, and we define all profiles to be comparable. Up to strategic equivalence, the game is defined by the \emph{payoff differences} between comparable profiles~\citep{candogan2011flows}. We store this in a matrix we call the \emph{weight matrix} $W$ of the game.

\begin{defn}
    Let $M$ be a zero-sum game, and let $p$ and $q$ be comparable profiles. If $M$ is a symmetric game, then profiles and strategies are the same, and we define $W_{p,q}$ to be the same as $M_{p,q}$. If $M$ is non-symmetric, then $p=(p_1,p_2)$ and $q=(q_1,q_2)$ and
    \[
    W_{p,q} =
    \begin{cases}
    M_{p_1,p_2} - M_{q_1,q_2} & p_2 = q_2\ \text{(\textit{i.e.}, $p$ and $q$ are $1$-comparable)} \\
    M_{q_1,q_2} - M_{p_1,p_2} & p_1 = q_1\ \text{(\textit{i.e.}, $p$ and $q$ are $2$-comparable)}
    \end{cases}
    \]
    If $p$ and $q$ are not comparable, then $W_{p,q}$ \emph{is undefined}.
\end{defn}
We deliberately leave the payoff differences between incomparable profiles undefined, so that it is clear to the reader that we will only reference $W$ when the associated profiles are comparable. Note that for any comparable profiles $p$ and $q$, $W_{p,q} = -W_{q,p}$. 
\begin{defn} \label{def: fundamental graph}
    Let $M$ be a zero-sum game. The \emph{preference graph} of $M$ is the graph whose nodes are the profiles of the game and where there is an arc $\arc{p}{q}$ between profiles $p$ and $q$ if and only if they are comparable and $W_{p,q} \leq 0$ (equivalently, $W_{q,p} \geq 0$).
\end{defn}
While the definition is the same for symmetric and non-symmetric games, the resultant graphs are not the same, because the weight matrix is defined differently. In symmetric zero-sum games all profiles are comparable, so the preference graph is a \emph{tournament}\footnote{In the degenerate case where $W_{p,q} = 0$, there is a pair of arcs 0-weighted arcs between nodes $p$ and $q$, which makes the preference graph not, strictly speaking, a tournament, but this will cause no problems.} (a directed graph with an arc between every pair of nodes). In non-symmetric games the preference graph is the game's response graph~\citep{biggar2023graph}, which is never a tournament because not all profiles are comparable. We think of the entries in the weight matrix as being weights on the associated arc, as in $\arc[|W_{p,q}|]{p}{q}$.
\subsection{Dynamical Systems and the Replicator} \label{sec: dynamical systems}
The replicator dynamic is a continuous-time dynamical system~\citep{sandholm2010population}, defined by an ordinary differential equation. Let $x\in\Delta(S_1)$ and $y\in\Delta(S_2)$ be mixed strategies, and let $s\in S_1$ and $t\in S_2$ be pure strategies. Then, for a (non-symmetric) zero-sum game $M$ we have
\begin{defn}[Non-Symmetric Zero-Sum Replicator Equation] \label{def: nonsymmetric replicator}
\begin{align*}
    \dot x_s &= x_s ((M y)_s - x^T M y) \\
    \dot y_t &= -y_t ((M^T x)_t - x^T M y)
\end{align*}
\end{defn}
In a symmetric game $M$, the replicator is defined a similar way:
\begin{equation*}
    \dot x_s = x_s ((M x)_s - x^T M x)
\end{equation*}
If $M$ is also zero-sum, then $x^T M x = 0$ (by anti-symmetry), and this reduces to
\begin{defn}[Symmetric Zero-Sum Replicator Equation] \label{def: symmetric replicator}
\[
    \dot x_s = x_s (M x)_s
\]
\end{defn}

Note that the symmetric and non-symmetric replicator, while similar, are distinct equations with different properties. The solutions to these equations define a \emph{flow}~\citep{sandholm2010population} on the strategy space of the game, which is a function $\phi:X\times \real\to X$ that is a continuous group action of the reals on $X$. We call these the symmetric or non-symmetric \emph{replicator flow}, respectively. The forward orbit of the flow from a given point is called a \emph{trajectory} of the system. A set of points $Y$ is called \emph{invariant} under $\phi$ if $\phi(Y,t) = Y$ for any $t\in\real$.

\begin{defn} \label{def: attractor}
Let $X$ be a compact space, $\phi$ a flow on $X$ and $A$ a compact subspace of $X$. If there is a neighbourhood $U$ of $A$ such that \[
\lim_{t\to\infty} \sup_{x\in U} \inf_{y \in A} \dist(\phi(x,t), y) = 0
\]
under any metric $\dist$, then we say $A$ is \emph{asymptotically stable}. If $A$ is also invariant under $\phi$, we call it an \emph{attracting set}. An attracting set which contains no smaller attracting sets is called an \emph{attractor}.
\end{defn}
There are some differences in terminology in the literature to be wary of. What we call attracting sets are sometimes called attractors~\citep{conley1978isolated,sandholm2010population,biggar2023replicator}, in which case what we call an attractor is a \emph{minimal attractor}. Otherwise, our definition is the same as~\citet{sandholm2010population,biggar2023replicator}.
The set of points which approach an attractor $A$ in the limit $t\to\infty$ is called the attractor's \emph{basin of attraction}. We call an attractor \emph{global} if its basin of attraction includes all points in $\intr(X)$~\citep{hofbauer2003evolutionary}.

\section{The Attractor of the Replicator}

In this section we prove Theorem~\ref{attractor characterisation}, which characterises the attractor of the replicator in zero-sum games. We begin by noting that any mixed profile $x$ naturally defines a distribution over profiles, with $x_p$ denoting the mass on a profile $p$. If the game is symmetric, profiles and strategies are the same and so this is trivial. In a non-symmetric game, $x=(x_1,x_2)$ is a pair of mixed strategies, and the distribution is defined by the \emph{product}, with the mass on a profile $p = (p_1,p_2)$ defined by
\begin{equation} \label{x sub p}
x_p := {x_1}_{p_1} {x_2}_{p_2}
\end{equation}
This distribution over profiles is used to define an important concept: the \emph{content} of a set of profiles.
\begin{defn}[\citeauthor{biggar2023replicator}] \label{def: content}
Let $H$ be a set of profiles in a game. The \emph{content} of $H$, denoted $\content(H)$, is the set of all mixed profiles $x$ where all profiles in the support of $x$ are in $H$.
\end{defn}

Equivalently, $x\in\content(H)$ if and only if $x_H := \sum_{h\in H} x_h = 1$, that is, $x$ defines a distribution whose mass is entirely distributed over profiles in $H$. The content is a union of subgames, and so is an invariant set under the replicator~\cite{sandholm2010population}. An example is shown in Figure~\ref{fig:outer diamond}.
We will show that the unique global attractor is the content of the unique sink component of the preference graph. Uniqueness follows from graph structure, established originally in~\cite{biggar2023graph}.
\begin{lem}[Uniqueness] \label{lem: uniqueness}
    The preference graph of a zero-sum game has a unique sink component.
\end{lem}
The proof of Lemma~\ref{lem: uniqueness} can be found in the appendix. Now we can prove our main theorem.

\begin{thm}[The Attractor of the Replicator] \label{attractor characterisation}
In a (symmetric or non-symmetric) zero-sum game $M$, the content of the unique sink component $H$ of its preference graph is the unique global attractor of the (respectively symmetric or non-symmetric) replicator dynamic.
\end{thm}
\begin{proof}
    Proving Theorem~\ref{attractor characterisation} requires showing (i) $\content(H)$ is an invariant set (\emph{invariance}), (ii) every attracting set contains $\content(H)$ (\emph{minimality}) and (iii) $\content(H)$ is asymptotically stable and its basin of attraction contains $\intr(X)$ (\emph{global asymptotic stability}). We establish (i) and (ii) in Lemma~\ref{lem: lower bound}.
    \begin{lem}[Invariance and Minimality] \label{lem: lower bound}
    If $H$ is the sink component of the preference graph of a (symmetric or non-symmetric) zero-sum game $M$, then $\content(H)$ is invariant under the replicator. Further, for any attracting set $A$, $\content(H)\subseteq A$.
    \end{lem}
    Lemma~\ref{lem: lower bound} is a streamlined presentation of existing results~\citep[particularly Theorem~5.2 of][]{biggar2023replicator}, so we defer its proof to the appendix. The challenge and main contribution of Theorem~\ref{attractor characterisation} lies in (iii): showing global asymptotic stability of $\content(H)$. We do this by demonstrating that $x_H$, the total mass on the sink component $H$, increases over time (Lemmas~\ref{lem: symmetric lyapunov} and~\ref{lem: two-pop lyapunov}). That is, $\dot x_H = \ddt x_H = \sum_{h\in H} \ddt x_h = \sum_{h\in H}\dot x_h > 0$, for any $x\in\intr(X)$. The function $x_H$ is a natural choice, because $x_H$ is uniformly continuous, bounded in $[0,1]$ and $x_H = 1$ if and only if $x\in\content(H)$, so $x_H$ can be thought of as a metric for the distance between $x$ and the content. Showing $\dot x_H > 0$ requires different arguments for the symmetric and non-symmetric cases.  The symmetric case is straightforward, and we complete it in Lemma~\ref{lem: symmetric lyapunov} below.

\begin{lem} \label{lem: symmetric lyapunov}
In a symmetric zero-sum game $M$, under the symmetric replicator (Definition~\ref{def: symmetric replicator}), if $x_H\in (0,1)$ then $\dot x_H > 0$.
\end{lem}
\begin{proof}
First, $\dot x_H = \sum_{h\in H} x_h (M x)_h$. As $M$ is anti-symmetric, for any profiles $h$ and $q$ we have $x_h x_q M_{h,q} + x_q x_h M_{q,h} = 0$. Hence,
\[
\dot x_H = \sum_{q\in S}\sum_{h\in H} x_q x_h M_{h,q} = \sum_{q\in H}\sum_{h\in H} x_q x_h M_{h,q} + \sum_{q\not\in H}\sum_{h\in H} x_q x_h M_{h,q} = \sum_{q\not\in H}\sum_{h\in H} x_q x_h M_{h,q}
\]
For any $q\notin H$ and $h\in H$, as $H$ is a sink component, there is an arc $\arc{q}{h}$ in the preference graph with strictly positive weight. That is, $M_{h,q} > 0$. Hence the summands are nonnegative. Finally, because $x_H \in (0,1)$, there must exist a $q\not\in H$ and an $h\in H$ such that $x_qx_h > 0$. Thus, $\dot x_H > 0$.
\end{proof}

The more challenging case of the proof is showing that $\dot x_H > 0$ in non-symmetric games (Lemma~\ref{lem: two-pop lyapunov}). We solve this using Theorem~\ref{thm: symmetrisation}, which establishes that we can embed the flow of the non-symmetric replicator on our non-symmetric game $M$ into the flow of the symmetric replicator on a larger symmetric game, $M$'s \emph{von Neumann symmetrisation}.
\begin{defn}[Von Neumann Symmetrisation]\label{def:symised game}
Let $M$ be a zero-sum game, with $M\in \real^{n\times m}$. The \emph{von Neumann symmetrisation} of $M$, written $\sym_M$, is defined as the following $\real^{(nm)\times (nm)}$ matrix, which we index by profiles $p = (p_1,p_2)$ and $q = (q_1,q_2)$:
\[
(\sym_M)_{p,q} = M_{p_1,q_2} - M_{q_1,p_2}.
\]
\end{defn}
The strategy space of $\sym_M$ is $S_1\times S_2$. That is, the \emph{strategy profiles} of the original game $M$ become \emph{strategies} of the von Neumann symmetrisation. Mixed profiles likewise become \emph{mixed strategies}, using the production distribution as in equation~\eqref{x sub p}: $x_p := {x_1}_{p_1} {x_2}_{p_2}$.

This construction appeared first in~\citeauthor{gale20167}, who attributed it to John von Neumann, hence the name. It possesses an intuitive explanation: it results from holding two plays of the original game $M$, with each player taking the role of the row player 1 in one play and player 2 in the other~\citep{gale20167}. The von Neumann symmetrisation has two important properties: (1) it is anti-symmetric, and so is a symmetric zero-sum game, and (2) it can be viewed as an extension of the weight matrix $W$, previously only defined on comparable profiles, to a relation over all profiles. That is, if $p$ and $q$ are comparable then $(\sym_M)_{p,q} = W_{p,q}$. More generally,
\begin{lem} \label{lemma: symmetry weights}
Let $p = (p_1,p_2)$ and $q = (q_1,q_2)$ be profiles. Then:
\[ 
(\sym_M)_{p,q} = W_{p,(p_1,q_2)} + W_{p,(q_1,p_2)} = W_{(p_1,q_2),q} + W_{(q_1,p_2),q}.
\]
\end{lem}
 In Theorem~\ref{thm: symmetrisation}, we establish a novel link between the replicator flow on a zero-sum game and its von Neumann symmetrisation.

\begin{thm} [Symmetrising the Replicator Dynamic] \label{thm: symmetrisation}
Let $M$ be a non-symmetric zero-sum game. Let $x = (x_1,x_2)$ be a mixed profile and $p = (p_1,p_2)$ a pure profile. Write $x_p := {x_1}_{p_1} {x_2}_{p_2}$ as in equation~\eqref{x sub p}. Then, under the non-symmetric replicator (Definition~\ref{def: nonsymmetric replicator}),
\[ \dot x_p = x_p (\sym_M x)_p. \]
\end{thm}
In English, Theorem~\ref{thm: symmetrisation} states that the flow of the non-symmetric replicator (Definition~\ref{def: nonsymmetric replicator}) on $M$ embeds (on the subspace of product distributions) in the flow of the symmetric replicator (Definition~\ref{def: symmetric replicator}) on $\sym_M$. The proofs of Theorem~\ref{thm: symmetrisation} and Lemma~\ref{lemma: symmetry weights} are straightforward and can be found in the appendix. 
Using Theorem~\ref{thm: symmetrisation}, we can complete the proof of Theorem~\ref{attractor characterisation} on non-symmetric games by showing that $x_H$ is increasing in $x_H\in (0,1)$.

\begin{lem} \label{lem: two-pop lyapunov}
    In a non-symmetric zero-sum game $M$, under the non-symmetric replicator (Definition~\ref{def: nonsymmetric replicator}), if $x_H\in (0,1)$ then $\dot x_H$ > 0.
\end{lem}
\begin{proof}
If $h=(h_1,h_2)$ is a pure profile, then by Theorem~\ref{thm: symmetrisation}, $\dot x_h = \ddt ({x_1}_{h_1}{x_2}_{h_2})= x_h (\sym_M x)_h $. Because $\sym_M$ is symmetric, by the same argument as in Lemma~~\ref{lem: symmetric lyapunov}, we can show that
    \[
\dot x_H = \sum_{q\not\in H}\sum_{h\in H} x_q x_h (\sym_M)_{h,q}.
\]
Now pick some profiles $q\not\in H$ and $h\in H$ with $x_qx_h > 0$. As in Lemma~\ref{lem: symmetric lyapunov}, because $x_H \in (0, 1)$, at least one such pair exist. We will show that the sum above is strictly positive. Firstly, observe that if $q$ and $h$ are comparable, the arc $\arc{q}{h}$ of the preference graph of $M$ goes from $q$ to $h$ with strictly positive weight, because $q\not\in H$ and $h\in H$, and so $(\sym_M)_{h,q} = W_{h,q} > 0$. Otherwise, suppose  that $q = (q_1,q_2)$ and $h = (h_1,h_1)$ are not comparable, and let $a = (q_1,h_2)$ and $b = (h_1,q_2)$. We have the following three cases:
\begin{enumerate}
    \item $a,b\in H$. Then the arcs $\arc{q}{b}$ and $\arc{q}{a}$ are directed towards $a$ and $b$ because $q$ is outside the sink component $H$. By Lemma~\ref{lemma: symmetry weights}, $(\sym_M)_{h,q} = W_{a,q} + W_{b,q} > 0$.

    \item $a,b\not\in H$. 
    Then the arcs $\arc{a}{h}$ and $\arc{b}{h}$ are directed towards $h$ because $a$ and $b$ are outside the sink component. By Lemma~\ref{lemma: symmetry weights}, $(\sym_M)_{h,q} = W_{h,a} + W_{h,b} > 0$.

    \item $a\in H$, $b\not\in H$ (the case $b\in H$, $a\not\in H$ is identical). The sum $\sum_{q\not\in H}\sum_{h\in H} x_q x_h (\sym_M)_{h,q}$ includes the terms $x_qx_h (\sym_M)_{h,q}$ and $x_bx_a (\sym_M)_{a,b}$. However, $x_qx_h = ({x_1}_{q_1}{x_2}_{q_2})({x_1}_{h_1}{x_2}_{h_2}) = ({x_1}_{q_1}{x_2}_{h_2})({x_1}_{h_1}{x_2}_{q_2}) = x_ax_b$, and so $x_qx_h (\sym_M)_{h,q} + x_bx_a (\sym_M)_{a,b} = x_qx_h ((\sym_M)_{h,q} + (\sym_M)_{a,b}) = x_qx_h(W_{a,q} + W_{b,q} + W_{q,b} + W_{h,b}) = x_qx_h(W_{a,q} + W_{h,b})$ by Lemma~\ref{lemma: symmetry weights}. Since $a$ and $h$ are inside the sink component and $q$ and $b$ are outside, and the arcs $\arc{q}{a}$ and $\arc{b}{h}$ must be directed into the component, so $W_{a,q} > 0$ and $W_{h,b} > 0$ and thus $x_qx_h (\sym_M)_{h,q} + x_bx_a (\sym_M)_{a,b} > 0$.
\end{enumerate}

Overall, we conclude that $\dot x_H > 0$.
\end{proof}
Finally, asymptotic stability of $\content(H)$ follows easily from the fact that $x_H$ is uniformly continuous. Pick any $0 < \alpha < \beta < 1$. Then for any $x$ with $x_H \in [\alpha,\beta]$, $\dot x_H > \epsilon$ for some $\epsilon$ (uniform continuity), so after some finite time $x_H > \beta$. Repeating this argument for any $\alpha,\beta$ shows that $\content(H)$ is asymptotically stable.
\end{proof}

\section{Chain Recurrence and Nash Equilibria} \label{sec: nash}

In this section we discuss some consequences of Theorem~\ref{attractor characterisation}. The first concerns chain recurrence, which is defined by $(\epsilon,T)$-chains.
\begin{defn}[$(\epsilon,T)$-Chains, \citeauthor{alongi2007recurrence}] \label{def: chain recurrence}
Let $\phi$ be a flow on a compact metric space $X$, with $x$ and $y$ in $X$. An \emph{$(\epsilon,T)$-chain} from $x$ to $y$ is a finite sequence of points $x_1,x_2,\dots,x_n$ with $x=x_1$ and $y=x_n$, and times $t_1,\dots,t_n \in [T,\infty)$ such that $\dist(\phi(x_i, t_i), x_{i+1}) < \epsilon$. If there is an $(\epsilon,T)$-chain from $x$ to $y$ for \emph{all} $\epsilon > 0$ and $T>0$ we say there is a \emph{pseudo-orbit} from $x$ to $y$.
\end{defn}
A point is called \emph{chain recurrent} if it has a pseudo-orbit to itself. Two points are \emph{chain equivalent} if there are pseudo-orbits between them in both directions, and equivalent chain recurrent points are grouped in topologically connected components called \emph{chain components}~\citep{alongi2007recurrence}. Reachability under pseudo-orbits provides an ordering on the chain components, and sink chain components are those which are minimal in this order. Sink chain components have been increasingly studied in algorithmic game theory~\citep{papadimitriou2016nash,papadimitriou2018nash,papadimitriou2019game}. A connection between the preference graph and chain components of the replicator was demonstrated by~\cite{biggar2023replicator}, who proved that replicator sink chain components always exist. However, sink chain components have not generally been characterised.  Attractors, when they exist, are sink chain components, so our results present the first characterisation of sink chain components of zero-sum games.
\begin{lem}[Folklore, see the appendix] \label{attractors are sink chain components}
    In any flow, every attractor is a sink chain component.
\end{lem}
\begin{lem} \label{sink chain components}
    The content of the sink component of the preference graph is the unique sink chain component of a zero-sum game.
\end{lem}
Lemma~\ref{sink chain components} follows directly from Theorem~\ref{attractor characterisation} combined with Lemma~\ref{attractors are sink chain components}. The proof of both lemmas can be found in the appendix. In zero-sum games, previous state-of-the-art results on chain recurrence made use of the Nash equilibrium~\citep{piliouras2014optimization,papadimitriou2016nash,mertikopoulos2018cycles}. This line of inquiry established first (1) that all points in the essential subgame are contained within the sink chain component. 
Secondly (2), if the essential subgame is not the whole game, all interior starting points converge to the essential subgame. This suggests a connection between chain recurrence and equilibria in zero-sum games. Lemma~\ref{sink chain components} comes to a different conclusion; it proves that sink chain components are characterised solely by the preference graph. Lemma~\ref{lem: preference nash} resolves this seeming discrepancy: 
chain components are determined by the preference graph, but the presence of an equilibrium in a subgame forces some structure on the induced preference graph of that subgame---in particular, it must be strongly connected and contained within the sink component of the whole game's preference graph.

\begin{lem} \label{lem: preference nash}
    In a zero-sum game, any subgame with a Nash equilibrium in its interior must (i) be contained within the sink component of the preference graph and (ii) the subgraph which this subgame induces in the preference graph must be strongly connected.
\end{lem}
\begin{proof}
(i) Assume the attractor is not the whole game, in which case it is on the boundary. By Theorem~\ref{attractor characterisation}, all interior points are in the basin of attraction, and by Theorem~3.4 of \cite{piliouras2014optimization}, trajectories starting from interior points converge to the essential subgame in the limit, and so the essential subgame must be within the attractor. (ii) Each subgame is independent, so we can assume we are working with the whole game and the equilibrium $x$ is fully-mixed. For contradiction, assume the preference graph is not strongly connected. Consequently, there is an attractor $A$ on the boundary of the strategy space (Theorem~\ref{attractor characterisation}). The point $x$ is in the basin of attraction of $A$, so $x$ converges to $A$, which contradicts the fact that $x$ is a Nash equilibrium, which are fixed points under the replicator~\citep{sandholm2010population}.
\end{proof}
This Lemma rephrases our understanding of equilibria and chain recurrence, highlighting the key role of the preference graph. As an example, Lemma~\ref{sink chain components} implies that sink chain component is the whole game \emph{if and only if} the preference graph is strongly connected. When a fully-mixed equilibrium exists, the preference graph must be connected, and so the sink chain component is the whole game. Most interestingly, even though we prove Lemma~\ref{lem: preference nash} using the replicator dynamic, this lemma is a \emph{purely game-theoretic result} which only relates equilibria and the preference graph. 
This connection between the preference graph and equilibria is useful for analysing games. For instance, examining Figure~\ref{fig:outer diamond} we find the Nash equilibrium is $((0,0.5,0.5),(0,0.5,0.5))$ (both players play $b$ and $c$ half the time), which has support $\{b,c\}\times \{b,c\}$. As Lemma~\ref{lem: preference nash} predicts, this is within the attractor and its induced preference graph (a 4-cycle) is strongly connected.

\section{Conclusions and Future Work}

In this paper we gave the first characterisation of the unique attractor of the replicator in zero-sum games, thereby describing the long-term behaviour of the dynamic in these games. As a secondary result, we have demonstrated the importance of the preference graph as a tool for analysing game dynamics. In particular, in the long run, only the \emph{preferences} dictates the outcome of the game. This surprising conclusion has potentially significant consequences for modelling strategic interactions. In game theory---especially evolutionary game theory---precise knowledge of utilities is difficult to achieve; our results show that modelling only the preferences for each player is sufficient to characterise the attractor of the replicator dynamic.

An important goal for future research is to characterise the attractors of the replicator in \emph{all} games, not just zero-sum ones. Currently, the attractors of the replicator have been characterised in only a few classes of games, such as potential games and $2\times n$ games without dominated strategies~\citep{biggar2023replicator}, as well as some individual games, like the Asymmetric Cyclic Matching Pennies game from \citep{kleinberg2011beyond}. Our paper adds zero-sum games to this list of solved classes.

\bibliographystyle{plainnat}
\bibliography{refs}

\appendix

\section{Proofs}

\begin{lem}[Lemma~\ref{lem: uniqueness}]
    The preference graph of a zero-sum game has a unique sink component.
\end{lem}
\begin{proof}
    In the non-symmetric case, the preference graph is the response graph, and the result follows from Theorem~4.10 of \cite{biggar2023graph}. In the symmetric case, the preference graph is a tournament, and all tournaments have one sink component, as they are orientations of complete graphs.
\end{proof}

\begin{lem}[Lemma~\ref{lem: lower bound}]
    If $H$ is the sink component of the preference graph of a (symmetric or non-symmetric) zero-sum game $M$, then $\content(H)$ is invariant under the replicator. Further, for any attracting set $A$, $\content(H)\subseteq A$.
\end{lem}
\begin{proof}
    This proof is largely the same as Theorem~5.2 of~\cite{biggar2023replicator}, with the addition of the symmetric case.
    (\emph{Invariance}:) Observe that if $x\in \content(H)$, then the support of $x$ is contained in $H$, and because all mixed profiles in the subgame $\Delta(\supp(x))$ have the same support, $\Delta(\supp(x))\subseteq\content(H)$. It follows that $\content(H)$ is a union of subgames. By Theorem~5.4.7 of \cite{sandholm2010population}, all subgames are invariant sets under the replicator, and unions of invariant sets are invariant.

    ($\content(H)\subseteq A$:) By Theorems~9.1.2 and 9.1.6 of \cite{sandholm2010population}, no asymptotically stable set can exist in the interior of the strategy space of a symmetric or non-symmetric zero-sum game. Subgames have the same properties as the whole game under the replicator, so the same is true of all subgames. Dually, no repelling set can exist in the interior of any subgame.

    (\emph{Claim}: every attracting set contains a profile.)
    This follows by induction, using the fact that the replicator dynamic on a subgame has the same properties as on the whole game. In the whole game, an asymptotically stable set intersects the boundary. This intersection with the boundary must also be asymptotically stable in any subgame it intersects on the boundary, and so it intersects the boundary of this smaller subgame, and so on. We conclude that such a set contains a pure profile, the smallest possible subgame. This claim generalises to other dynamics---see \cite{vlatakis2020no}, Theorem~4.5.

    (\emph{Claim}: every attracting set contains all profiles in $H$.)
    An arc $\arc{p}{q}$ of the preference graph is also a subgame, where only the profiles $p$ and $q$ are in the support. The (symmetric or non-symmetric) replicator reduces to $\dot x_p = x_p(1-x_p) W_{q,p}$ on this subgame, where $W_{q,p} \geq 0$ (for the player for which these profiles are comparable). If $p$ is contained in an asymptotically stable set, then $ q$ must also be contained in this set, because points near $p$ move to $q$ along this arc. We know that asymptotically stable sets contain a pure profile---by this argument we deduce that they contain all pure profiles reachable from that one in the preference graph. Such a set of profiles always contains the sink component $H$.

    (\emph{Claim}: every attracting set contains $\content(H)$.)
    Let $Y$ be a subgame, where the pure profiles in $Y$ are in an attracting set $A$. If $Y$ is a pure profile, then all mixed profiles in $Y$ are in the set, trivially. Now suppose for induction that all points on the boundary of $Y$ are in $A$. Suppose for contradiction that there is a point $x\in\intr(Y)$ that is not in $A$. The set $A\cap Y$ is attracting in $Y$, and the boundary is contained in $A\cap Y$, but this means that the dual repelling set of $A\cap Y$ is contained in the interior of $Y$, but no such sets can be contained in the interior. Hence all of $Y$ is within $A$. By induction on subgames, we find that all of $\content(H)$ is within every attracting set.
\end{proof}

\begin{lem} \label{lem: antisymmetry}
$\sym_M$ is anti-symmetric.
\end{lem}
\begin{proof}
For $p = (p_1,p_2)$ and $q = (q_1,q_2)$, $(\sym_M)_{p,q} = M_{p_1,q_2} - M_{q_1,p_2} = -(M_{q_1,p_2} - M_{p_1,q_2}) = -(\sym_M)_{q,p}$.
\end{proof}
\begin{lem}[Lemma~\ref{lemma: symmetry weights}]
Let $p = (p_1,p_2)$ and $q = (q_1,q_2)$ be profiles. Then:
\[ 
(\sym_M)_{p,q} = W_{p,(p_1,q_2)} + W_{p,(q_1,p_2)} = W_{(p_1,q_2),q} + W_{(q_1,p_2),q}.
\]
\end{lem}
\begin{proof}
    \begin{align*}
        W_{p,(p_1,q_2)} + W_{p,(q_1,p_2)} &= (M_{p_1,q_2}-M_{p_1,p_2}) + (M_{p_1,p_2} - M_{q_1,p_2}) \\
        &= M_{p_1,q_2} - M_{q_1,p_2} = (\sym_M)_{p,q} \quad \text{and}\\
        W_{(p_1,q_2),q} + W_{(q_1,p_2),q} &= (M_{p_1,q_2} -M_{q_1,q_2}) + (M_{q_1,q_2} - M_{q_1,p_2}) \\
        &= M_{p_1,q_2} - M_{q_1,p_2} = (\sym_M)_{p,q}
    \end{align*}
\end{proof}

\begin{thm}[Theorem~\ref{thm: symmetrisation}]
Let $M$ be a non-symmetric zero-sum game. Let $x = (x_1,x_2)$ be a mixed profile and $p = (p_1,p_2)$ a pure profile. Write $x_p := {x_1}_{p_1} {x_2}_{p_2}$ as in equation~\eqref{x sub p}. Then, under the non-symmetric replicator (Definition~\ref{def: nonsymmetric replicator}),
\[ \dot x_p = x_p (\sym_M x)_p. \]
\end{thm}
\begin{proof}
The two-population replicator dynamic (written for player 1, the player 2 case is similar) is equivalent to
\begin{align*}
    \dot {x_1}_s &= {x_1}_s ((M x_2)_s - x_1^T M x_2) && (\text{Definition~\ref{def: nonsymmetric replicator}})\\
    \dot {x_1}_s/{x_1}_s &= \sum_{t\in S_1} {x_1}_t ((M x_2)_s - (M x_2)_t) && (\text{as}\ \sum_{t\in S_1} {x_1}_t = 1)\\
    &= \sum_{t \in S_1} {x_1}_t \sum_{r\in S_2}  {x_2}_r \left (M_{s,r} - M_{t,r} \right ) \\
    &= \sum_{t\in S_1} \sum_{r \in S_2} {x_1}_t {x_2}_r \left (M_{s,r} - M_{t,r} \right ) \\
    &= \sum_{p=(p_1,p_2)\in S_1\times S_2} x_p \left (M_{s,p_2} - M_{p_1,p_2} \right ) && (\text{relabelling})
\end{align*}

Now we observe that for $p = (p_1,p_2)$,
\begin{align*}
    \dot x_p &= \ddt ({x_1}_{p_1}{x_2}_{p_2}) \\
    &= ({x_1}_{p_1}{x_2}_{p_2}) (\frac{\dot {x_1}_{p_1}}{{x_1}_{p_1}} + \frac{\dot {x_2}_{p_2}}{{x_2}_{p_2}}) && \text{(product rule)} \\
    &= x_p \Bigg(\sum_{q\in S_1\times S_2} x_q (M_{p_1,q_2} - M_{q_1,q_2}) \\
    &\qquad\ + \sum_{q\in S_1\times S_2} x_q (M_{q_1,q_2} -M_{q_1,p_2})\Bigg ) && \text{(by above)}\\
    &= x_p \sum_{q\in S_1\times S_2} x_q \left (M_{p_1,q_2} - M_{q_1,q_2} + M_{q_1,q_2} - M_{q_1,p_2} \right )\\
    &=  x_p \sum_{q\in S_1\times S_2} x_q (\sym_M)_{p,q} && \text{(Definition~\ref{def:symised game})}\\
    &=  x_p (\sym_M x)_p.
\end{align*}
\end{proof}

\begin{lem}[Lemma~\ref{attractors are sink chain components}]
    In any flow, every attractor is a sink chain component.
\end{lem}
\begin{proof}
    We first show that all points in the attractor are chain recurrent. Attracting sets are closed under intersection~\citep{kalies2021lattice,biggar2023replicator}, so an attractor cannot overlap any other attracting set---that would contradict the minimality of the attractor. \citet{conley1978isolated} showed that points are chain recurrent if, for each attracting set $A$, the point is contained in either $A$ or its dual repelling set $A^*$. The attractor is compact, invariant, and no attracting set overlaps it, so an attracting or repelling set must contain all points in the attractor. Hence all points are chain recurrent.

    Pseudo-orbits cannot leave attracting sets \citep{akin1984evolutionary}. Consequently, no point outside the attractor is chain equivalent to a point inside it, and all points in the attractor are chain equivalent, so it is a chain component. It is a \emph{sink} chain component because no pseudo-orbits leave the set.
\end{proof}

\begin{corol}
    The content of the sink component of the preference graph is the unique sink chain component of a zero-sum game.
\end{corol}
\begin{proof}
    By Lemma~4.2 and Theorem~3.3, the content of the sink component is a sink chain component. Uniqueness follows for the same reason as in Theorem~3.3: distinct sink chain components are disjoint, but every sink chain component contains the content~\citep{biggar2023replicator}.
\end{proof}

\end{document}